\documentclass{llncs}

  \usepackage{amsmath}
  \usepackage{amssymb}
  \usepackage[all]{xy} 
  \usepackage{rotating}
  \usepackage{fancybox}
  \usepackage{stmaryrd}

\newcommand{\C}{{\sf C\hspace*{-0.9ex}\rule{0.15ex}%
       {1.3ex}\hspace*{0.9ex}}}

\newcommand{\ket}[1]
           {\ensuremath{\vert #1 \rangle}}

\newcommand{\entangled}
           {\bf{E}}

\newcommand{\tensor}
           {\ensuremath{\otimes}}

\newcommand{\cnot}
           {\ensuremath{\mathfrak{Cnot}}}

\newcommand{\hadamard}
           {\ensuremath{\mathfrak{H}}}

\newcommand{\phase}
           {\ensuremath{\mathfrak{T}}}

\newcommand{\lambdaq}
           {\ensuremath{\lambda_{L}^{Q}}}

\newcommand{\abstr}[3]
           {\ensuremath{\lambda #1:#2.#3}}

\newcommand{\app}[2]
           {\ensuremath{(#1~~#2)}}

\newcommand{\true}
           {\mbox{{\bf 1}}}

\newcommand{\false}
           {\mbox{{\bf 0}}}
\newcommand{\couple}[1]
           {\ensuremath{\langle #1 \rangle}}

\newcommand{\proj}[2]
            {\ensuremath{\pi_{#1}{#2}}}

\newcommand{\letcouple}[3]
           {{\mbox {\sf let~}}\ensuremath{\couple{#1}=#2}
                   \mbox{{\sf ~in~}}\ensuremath{#3}}

\newcommand{\ifthls}[3]
           {\mbox{{\sf if~}}\ensuremath{#1}\mbox{{\sf ~then~}}
            \ensuremath{#2}\mbox{{\sf ~else~}}\ensuremath{#3}}

\newcommand{\meas}
           {\mbox{{\sf meas}}}

\newcommand{\newq}
           {\mbox{{\sf new}}}

\newcommand{\bits}
           {\ensuremath{{\bf B}}}

\newcommand{\qbits}
           {\ensuremath{{\bf B}^{\circ}}}

\newcommand{\inferencerule}[3]
           {\genfrac{}{}{}{0}{#1}{#2}#3}

\newcommand{\superpose}[2]
           {\stackrel{{\displaystyle #1}}{{\displaystyle #2}}}

\newcommand{\AXQ}
           {[VarQ]}

\newcommand{\AXC}
           {[AxQ]}

\newcommand{\AX}
           {[Var]}

\newcommand{\AXCT}
           {[AxT]}

\newcommand{\AXCF}
           {[AxF]}

\newcommand{\WEAK}
           {[Wkg]}

\newcommand{\ARI}
           {\ensuremath{[\to I]}}

\newcommand{\ARE}
           {\ensuremath{[\to E]}}

\newcommand{\ALI}
           {\ensuremath{[\to^{\circ} I]}}

\newcommand{\ALE}
           {\ensuremath{[\to^{\circ} E]}}

\newcommand{\TII}
           {\ensuremath{[\tensor I]}}

\newcommand{\IFI}
           {\ensuremath{[IF\ I]}}

\newcommand{\TIE}
           {\ensuremath{[\tensor E]}}

\newcommand{\BETA}
           {\ensuremath{[\beta]}}

\newcommand{\BETAV}
           {\ensuremath{[\beta V]}}

\newcommand{\APP}
           {\ensuremath{[App]}}

\newcommand{\APPC}
           {\ensuremath{[Apc]}}

\newcommand{\APPV}
           {\ensuremath{[Apv]}}

\newcommand{\IF}
           {\ensuremath{[IF]}}

\newcommand{\IFF}
           {\ensuremath{[IF/F]}}

\newcommand{\IFT}
           {\ensuremath{[IF/T]}}

\newcommand{\CPLL}
           {\ensuremath{[LFT]}}

\newcommand{\CPLR}
           {\ensuremath{[RGT]}}

\newcommand{\PHASE}
           {\ensuremath{[PHS]}}

\newcommand{\HADA}
           {\ensuremath{[HDR]}}

\newcommand{\MEST}
           {\ensuremath{[MET]}}

\newcommand{\MESF}
           {\ensuremath{[MEF]}}

\newcommand{\CNOT}
           {\ensuremath{[CNO]}}

\newcommand{\redprob}[1]
           {\ensuremath{\to_{#1}}}

\newcommand{\redprobabstr}
           {\ensuremath{\to_{{\cal A}}}}

\newcommand{\et}
           {~~~~}

\newcommand{\assertion}[3]
           {\ensuremath{\{#1\} #2 \{#3\}}}

\newcommand{\entangle}
           {\ensuremath{\leftrightarrow}}

\newcommand{\pure}
           {\ensuremath{\|}}

\newcommand{\evaluation}[3]
           {\ensuremath{#1 \bullet #2 = #3}}

\newcommand{\safe}[1]
           {}

\newcommand{\relentangle}
          {\ensuremath{{\mathcal R}}}

\newcommand{\purestate}
           {\ensuremath{{\mathcal P}}}

\newcommand{\PHASEABS}
           {\ensuremath{[PHS_{{\cal A}}]}}

\newcommand{\HADAABS}
           {\ensuremath{[HDR_{{\cal A}}]}}

\newcommand{\MESABS}
           {\ensuremath{[MET_{{\cal A}}]}}

\newcommand{\NEWABS}
           {\ensuremath{[NEW_{{\cal A}}]}}

\newcommand{\CNOTABSONE}
           {\ensuremath{[CNO1_{{\cal A}}]}}

\newcommand{\CNOTABSZERO}
           {\ensuremath{[CNO0_{{\cal A}}]}}

\newcommand{\random}
           {\ensuremath{\frac{\true}{\false}}}

\newcommand{\judge}[5]
           {\{#1\} #2 :^{#3} #4 \{#5\}}

\newcommand{\JDGLET}
           {[\pi_{J}]}

\newcommand{\JDGLOG}
           {[LOG_{J}]}

\newcommand{\JDGCNOTONE}
           {[CNOT1_{J}]}

\newcommand{\JDGMEAS}
           {[MEAS_{J}]}

\newcommand{\JDGCNOTTWO}
           {[CNOT2_{J}]}

\newcommand{\JDGHAD}
           {[HAD_{J}]}

\newcommand{\JDGPHASE}
           {[PHASE_{J}]}

\newcommand{\JDGVAR}
           {[VAR_{J}]}
         
\newcommand{\JDGCONST}
           {[CONST_{J}]}

\newcommand{\JDGIF}
           {[IF_{J}]}

\newcommand{\JDGAPP}
           {[APP_{J}]}

\newcommand{\JDGABS}
           {[ABS_{J}]}

\newcommand{\JDGCPL}
           {[\times_{J}]}

\newcommand{\JDGTRANS}
           {[\implies_J]}

\newcommand{\JDGPROM}
           {[promote]}

\newcommand{\JDGIMPEL}
           {[\implies ELim]}

\newcommand{\JDGETEL}
           {[\wedge Elim]}

\newcommand{\JDGOUL}
           {[\vee L]}

\newcommand{\JDGETR}
           {[\wedge R]}
\newcommand{\JDGEXL}
           {[\exists L]}

\newcommand{\JDGFORR}
           {[\forall R]}

\newcommand{\PROJ}
           {[\pi]}

\newcommand{\equobs}
           {\ensuremath{\equiv}}

\newcommand{\progsem}[2]
           {\llbracket #1\rrbracket_#2}

\newcommand{\model}
           {\ensuremath{{\cal M}}}

\newcommand{\interpretation}
           {\ensuremath{{\cal I}}}

\newcommand{\abstrvalset}[1]
           {\ensuremath{\Xi}_{#1}}

\newcommand{\interpret}[1]
           {\ensuremath{[\!\!\vert #1 \vert\!\!]}}

\newcommand{\TTAX}
           {[TAsAx]}

\newcommand{\TTQ}
           {[TAsQ]}

\newcommand{\TTPROD}
           {[TAs\otimes]}

\newcommand{\TTPI}
           {[TAs\pi_{i}]}

\newcommand{\ATEN}
           {[TAs\entangle]}

\newcommand{\ATPU}
           {[TAs\pure]}

\newcommand{\ATEQ}
           {[TAs=]}

\newcommand{\ATNEG}
           {[TAs\neg]}

\newcommand{\ATAN}
           {[TAs\wedge]}

\newcommand{\ATOR}
           {[TAs\vee]}

\newcommand{\ATIMP}
           {[TAs\implies]}

\newcommand{\ATFORALL}
           {[TAs\forall]}

\newcommand{\ATEXISTS}
           {[TAs\exists]}

\newcommand{\ATHIGORD}
           {[TAsEV]}

\newcommand{\THAD}
           {[HAD]}

\newcommand{\TPHA}
           {[PHA]}

\newcommand{\TMEAS}
           {[MEAS]}

\newcommand{\TCNOT}
           {[CNOT]}

\newcommand{\logic}
           {\ensuremath{\vdahs_{L}}}

\newcommand{\vrai}
           {\ensuremath{\mathsf{T}}}

\newcommand{\faux}
           {\ensuremath{\mathsf{F}}}

\newcommand{\equivclass}[1]
           {\widetilde{#1}}

\begin{document}

\title{A logical analysis of entanglement and separability
       in quantum higher-order functions}
\author{F. Prost \and C. Zerrari}
 
\institute{ LIG\\
            46, av F\'elix Viallet,
            F-38031 Grenoble, France\\
            \email{Frederic.Prost@imag.fr}}

\maketitle

\begin{abstract}
  We present a logical separability analysis for a functional quantum
  computation language. This logic is inspired by previous works on logical
  analysis of aliasing for imperative functional programs. Both
  analyses share similarities notably because they are highly
  non-com\-pos\-it\-ional. Quantum setting is harder to deal with since
  it introduces non determinism and thus considerably modifies
  semantics and validity of logical assertions. This logic is the
  first proposal of entanglement/separability analysis dealing with a
  functional quantum programming language with higher-order functions.       
\end{abstract}

\section{Introduction}
\label{sec:introduction}

The aim of high level programming language is to provide a
sufficiently high level of abstraction in order both to avoid
unnecessary burden coming from technical details and to provide useful
mental guidelines for the programmer. Quantum computation
\cite{NielChu00} is still in its prime and quantum programing
languages remain in need for such abstractions. Functional quantum
programing languages have been proposed and offer ways to handle the
no-cloning axiom via linear $\lambda$-calculi
\cite{Ton04,SelVal05}. In \cite{AltGra05} is developed QML in which a
purely quantum control expression is introduced in order to represent
quantum superposition in programming terms. Another crucial ingredient
of quantum computation is the handling of entanglement of quantum states during
computation. Indeed without entanglement it is possible to efficiently
simulate quantum computations on a classical computer \cite{Vid03}. A
first step to deal with entanglement, and its dual: separability, has
been done in \cite{Perdr07} in which a type system is provided in order to
approximate the entanglement relation of an array of quantum bits.

 Quantum bits entanglement analysis shares some similarities with
variables name aliasing analysis.Indeed, aliasing analyzes are
complicated since an action on a variable of a given name may have 
repercussions on another variable having a different name. The same
kind of problems occur between two entangled quantum bits : if one
quantum bit is measured then the other one can be affected. In both
cases there is a compositionality issue: it is hard to state anything
about a program without any knowledge of its context. It seems therefore
sensible to try to adapt known aliasing analysis techniques to the
quantum setting. 

 In this paper we follow the idea developed in \cite{BerHonYos05} and
adapt it for entanglement/separability analysis in a functional
quantum programing language with higher order functions. The work of
\cite{BerHonYos05} has to be adapted in a non deterministic setting,
which is inherent of quantum computation, making the semantics and
soundness of the logic radically different. Moreover, our results
are a strict improvement over \cite{Perdr07} in which only first order
functions are considered.

\subsection{outline of the paper}

We first start by giving the definition of the dual problems of
entanglement and separability, together with quick reminders on
quantum computation, in section \ref{sec:separability}. Then, in
section \ref{sec:functional}, we present a functional quantum
computation language in section for which we define an entanglement
logic in section \ref{sec:entanglement}. Finally, we conclude in
section \ref{sec:conclusion}.

\section{Separability and Entanglement}
\label{sec:separability}

 A $n$ qubits register is represented by a normalized vector in a Hilbert 
$2^n$-dimension space that is the tensorial product of $n$ dimension $2$ 
Hilbert spaces on $\C^{2}$. Each 2 dimension subspace represents a
qubit. For a given vector, written $\ket{\varphi}$, qubits can be
either entangled or separable.

  \begin{definition}[Entanglement, Separability]
     \label{def:entangle}
     Consider $\ket{\varphi}$ a $n$ qubits register. $\varphi$ is
     separable if it is possible to partition the $n$ qubits in two
     non empty sets $A, B$, two states $\ket{\varphi_A}$ and
     $\ket{\varphi_B}$ describing $A$ and $B$ qubits, such that $\ket{\varphi}
     = \ket{\varphi_A} \otimes \ket{\varphi_B}$, where
     $\ket{\varphi_A}$ and $\ket{\varphi_B}$,otherwise it is said
     entangled. 

     By extension, two qubits $q,q'$ are separable if and only if
     there exists a partition $A, B$, two states  $\ket{\varphi_A}$ and
     $\ket{\varphi_B}$ describing $A$ and $B$ qubits,  such that $\ket{\varphi}
     = \ket{\varphi_A} \otimes \ket{\varphi_B}$, with $q \in A$ and $q' \in
     B$. Otherwise $q, q'$ are entangled. 
  \end{definition}

  \begin{definition}[Entanglement relation]
    \label{def:entangle_relation}
    Let a $n$ qubits register be represented by $\ket{\varphi}$. The
    entanglement relation of $\ket{\varphi}$,
    $\entangled(\ket{\varphi})$, over qubits of the register is
    defined  as follows: $(x,y) \in \entangled(\ket{\varphi})$ if and
    only if $x$ and $y$ are entangled. 
  \end{definition}

   The entanglement relation is an equivalence relation. It is indeed
   obviously symmetric and reflexive. It is transitive because if $(x,y)
   \in \entangled(\ket{\varphi})$
   and $(y,z)\in \entangled(\ket{\varphi})$. It is possible to find
   a partition $X,Z$ (with $x \in X$ and $z \in Z$) and
   $\ket{\varphi_X},\ket{\varphi_Z}$ such that $\ket{\varphi}=
   \ket{\varphi_X} \otimes \ket{\varphi_Z}$. $y$ is either in $X$ or
   $Y$ then either $(x,y)$ or $(y,z)$ is not in $
   \entangled(\ket{\varphi})$, thus the result by contradiction.

\section{$\lambdaq$ a functional quantum computing language}
\label{sec:functional}

We use a variant of Selinger and Valiron's $\lambda$-calculus
\cite{SelVal05} as programming language. Instead of considering
arbitrary unitary transformations we only consider three: quantum
phase $\phase$, Hadamard transformation \hadamard, and conditional not
\cnot. This restriction doest not make our language less general since
it forms a universal quantum gates set, see \cite{NielChu00}. It makes
entanglement analysis simpler. Indeed, only \cnot\ may create
entanglement. We also introduce another simplifications: since the
calculus may be linear only for quantum bits we do not use all the
linear artillery (bang, linear implications etc) but only check that
abstractions over quantum bits are linear. Moreover we suppose a fixed
number of quantum bits, therefore there are no $\newq$ operators
creating new quantum bits during computation. Indeed as shown in
\cite{Pro07} name generation creates nontrivial problems.

  Therefore, in the following we suppose the number of quantum bit
registers fixed although non specified and refer to it as $n$. 
  
  \subsection{Syntax and types}
 
  \begin{definition}[Terms and types]
  \label{def:terms}
    \lambdaq terms and types are inductively defined by:
    $$\begin{array}{rcl}
        M,N,P & ::= & x \mid q_{i} \mid \true \mid \false \mid \abstr{x}{M}{N} \mid \app{M}{N}  \\
              &     & \mid \true \mid \false \mid \couple{M,N}  
                                        \mid \proj{i}{N}
                                  \mid \ifthls{M}{N}{P} \mid \\ 
              &     &  \meas \mid 
                       \cnot \mid \hadamard \mid \phase\\ \\
              
        \sigma,\tau & ::= & \bits \mid \qbits \mid \sigma \to \tau  
                        \mid \sigma \tensor \tau \\
      \end{array}$$

      where $x$ denotes names of element of a countable set of
      variables. $q_i$, where $i \in \{1..n\}$ are constant names that
      are used as reference for a concrete quantum bit array. $\true,
      \false$ are standard boolean constant. $\proj{i}{N}$ with $i\in
      \{1,2\}$ is the projection operator. Terms of the third line
      are quantum primitives respectively for measure, quantum bit
      initialization and the three quantum gates Conditional not,
      Hadamard and phase.

      We only have two base types $\bits$ for bits and $\qbits$ for
      quantum bits, arrow and product types are standard ones.
  \end{definition}

 Note that if quantum bits are constants, there can be quantum bits
 variable in this $\lambdaq$. Consider for instance the following
 piece of code: $\app{\abstr{x}{M}}{\ifthls{M}{q_1}{q_2}}$. After
 reduction $x$ may eventually become either $q_1$ or $q_2$. We write
 $q$ without subscript to denote quantum bit variables. 

  \begin{definition}[Context and typing judgments]
  \label{def:typing}
     Contexts are inductively defined by:
     $$ \Gamma ::= \cdot \mid \Gamma, x:\sigma$$
     where $\sigma$ is not $\qbits$. 

     We define lists of quantum bits variable by:
     $$\Lambda ::= \cdot \mid \Lambda,q$$

     Typing judgments are of the form:
     $$\Gamma; \Lambda \vdash M : \sigma$$
     and shall be read as : under the typing context $\Gamma$, list of
     quantum bits variable $\Lambda$ , the term $M$ is well formed of type 
     $\sigma$. 
  \end{definition}

  As usual we require that typing contexts and lists are
unambiguous. It means that when we write $\Gamma,x:\sigma$
(resp. $\Lambda,q$) $x$ (resp. $q$) is implicitly supposed not to
appear in $\Gamma$ (resp. $\Lambda$). Similarly when we write
$\Gamma_{1}, \Gamma_{2}$ (resp. $\Lambda_{1}, \Lambda_{2}$) we
intend that $\Gamma_{1}$ and $\Gamma_{2}$ (resp. $\Lambda_{1}$ and
$\Lambda_{2}$) are disjoint contexts.

  Typing rules are the following :

  $$ \inferencerule{}
                   {\Gamma;\Lambda \vdash q_i:\qbits}
                   {\AXC}$$

  $$\begin{array}{ccc}
      \inferencerule{}
                    {\Gamma; q \vdash q:\qbits}
                    {\AXQ}  &\et \et &
      \inferencerule{}
                    {\Gamma,x:\sigma; \cdot \vdash x:\sigma}
                    {\AX} \\ \\
      \inferencerule{}
                    {\Gamma;\cdot \vdash \true : \bits}
                    {\AXCT} & &
      \inferencerule{}
                    {\Gamma;\cdot \vdash \false : \bits}
                    {\AXCF}\\
     \end{array}$$

     $$\inferencerule{\Gamma; \cdot \vdash M:\sigma}
                    {\Gamma,x:\tau; \cdot \vdash M:\sigma}
                    {\WEAK}$$

  $$\begin{array}{cc}
      \inferencerule{\Gamma, x:\sigma ; \Lambda \vdash M : \tau}
                    {\Gamma; \Lambda \vdash \abstr{x}{\sigma}{M} : 
                     \sigma \to \tau}
                    {\ARI}  &
      \inferencerule{\Gamma; \Lambda \vdash M: \sigma \to \tau \et 
                           \Gamma;\cdot \vdash N :\sigma}
                    {\Gamma;\Lambda \vdash \app{M}{N} : \tau}
                    {\ARE} \\ \\
     \inferencerule{\Gamma; \Lambda,q \vdash M : \tau}
                    {\Gamma; \Lambda \vdash \abstr{q}{\qbits}{M} : 
                           \qbits \to \tau}
                    {\ALI} &
      \inferencerule{\Gamma; \Lambda_{1} \vdash M: \sigma \to \tau \et 
                           \Gamma;\Lambda_{2} \vdash N :\sigma}
                    {\Gamma;\Lambda_1,\Lambda_{2} \vdash \app{M}{N} : \tau}
                    {\ALE} \\
    \end{array}$$

$$\begin{array}{c}
%
      \inferencerule{\Gamma; \Lambda_{1} \vdash M: \tau \et 
                           \Gamma;\Lambda_{2} \vdash N :\sigma}
                    {\Gamma;\Lambda_1,\Lambda_{2} \vdash 
                           \couple{M,N} : \tau \tensor \sigma}
                    {\TII} \\ \\
      \inferencerule{ \Gamma; \Lambda_{1} \vdash M : \bits \et 
                      \Gamma; \Lambda_{2} \vdash N:\tau \et 
                      \Gamma; \Lambda_{2} \vdash P:\tau}
                    {\Gamma; \Lambda_{1},\Lambda_2 \vdash \ifthls{M}{N}{P} : \tau}
                    {\IFI}\\ \\
      \inferencerule{\Gamma'; \Lambda \vdash M: \tau_1 \tensor \tau_2}
                    {\Gamma; \Lambda \vdash \proj{i}{M} :\tau_i}
                    {\TIE i} i\in \{1,2\} \\ \\
      \inferencerule{\Gamma;\Lambda \vdash M:\qbits}
                    {\Gamma;\Lambda \vdash \hadamard\ M :\qbits}
                    {\THAD} \\ \\
      \inferencerule{\Gamma;\Lambda \vdash M:\qbits}
                    {\Gamma;\Lambda \vdash \phase\ M:\qbits}
                    {\TPHA} \\ \\
      \inferencerule{\Gamma;\Lambda \vdash M:\qbits}
                    {\Gamma;\Lambda \vdash \meas\ M:\bits}
                    {\TMEAS}\\ \\
      \inferencerule{\Gamma;\Lambda \vdash M:\qbits \otimes \qbits}
                    {\Gamma;\Lambda \vdash \cnot\ M:\qbits \otimes \qbits}
                    {\TCNOT}
    \end{array}$$

Where in rule \TIE\, $\Gamma'=\Gamma,x:\sigma,y:\tau$ if $\sigma$ and
$\tau$ are not $\qbits$, $\Gamma'=\Gamma,x:\sigma$
(resp. $\Gamma,y:\tau$) if $\tau$ (resp. $\sigma$) is $\qbits$ and
$\sigma$ (resp. $\tau$) is not $\qbits$. $\Lambda_{2}'$ is build in a
symmetrical way, thus $\Lambda_{2}'$ is $\Lambda_{2}$ augmented with
variables $x$ or $y$ if and only if their type is $\qbits$.   

   $\lambdaq$ is a standard simply typed $\lambda$-calculus with two
base types which is linear for terms of type $\qbits$. Thus we ensure
the no-cloning property of quantum physics (e.g. \cite{NielChu00}).  

  \subsection{Operational semantics}
 
  Quantum particularities have strong implications in the design of a
quantum programming language. First, since quantum bits may be
entangled together it is not possible to textually represent
individual quantum bits as a normalized vector of $\C^{2}$. We use
$\ket{\true}$ and $\ket{\false}$ as base. Therefore, a quantum
program manipulating $n$ quantum bits is represented as a quantum
state of a Hilbert space $\C^{2n}$ and constants of type $\qbits$
are pointers to this quantum state. Moreover, quantum operators
modify this state introducing imperative actions. As a
consequence an evaluation order has to be set in order to keep some
kind of confluence. Moreover, $\lambdaq$ reductions are
probabilistic. Indeed, quantum mechanics properties induce an
inherent probabilistic behavior when measuring the state of a quantum
bit. 

\begin{definition}[$\lambdaq$ state]
  \label{def:lambdaqstate}
  Let $\Gamma;\Lambda \vdash M: \sigma$. A $\lambdaq$ state is a couple $[\ket{\varphi},M]$.where
$\ket{\varphi}$ is a normalized vector of $\C^{2n}$ Hilbert space and
$M$ a $\lambdaq$ term. 
\end{definition}  

  An example of $\lambdaq$ state of size $n=2$ is the following:
  $$[\ket{\varphi}, 
     \app{\abstr{q}{\qbits}
         {\ifthls{\app{\meas}{q_1}}
                {\true}
                {\app{\meas}{\app{\phase}{q}}}}}{q_2}]$$  
where $\ket{\varphi}=\frac{1}{\sqrt{2}}(\ket{\false}+\ket{\true}) \tensor 
(\frac{2}{3}\ket{\false}+ \frac{\sqrt{5}}{3}\ket{\true})$ $q_1$ is the
quantum bit denoted by $\frac{1}{\sqrt{2}}(\ket{\false}+\ket{\true})$
and $q_2$ the one represented by $\frac{2}{3}\ket{\false}+
\frac{\sqrt{5}}{3}\ket{\true}$. 

  We consider call by value reduction rules. Values are defined as
usual.   

\begin{definition}[Values]
\label{def:values}
   Values of $\lambdaq$ are inductively defined by:
   $$U,V::= x \mid \true \mid \false\mid q_{i} \mid
   \abstr{x}{\sigma}{M} \mid \couple{V,V}\mid \app{F}{x}$$

   Where $F$ is one of the following operators $ \pi_i, \cnot, \phase,
   \hadamard, \meas$
\end{definition}

  We can now define probabilistic reduction rules. We only mention
probabilities to be accurate although we are not going to investigate
any related problems in this paper (we do not consider confluence
problems etc.). 

\begin{definition}[Quantum reductions]
\label{def:quantumreduction}
  We define a probabilistic reduction between $\lambdaq$ states 
as: 
  $$[\ket{\varphi},M] \redprob{p} [\ket{\varphi'},M']$$
That has to be red $[\ket{\varphi},M]$ reduces to 
$[\ket{\varphi'},M']$ with probability $p$.

Reduction rules are the following:
$$\begin{array}{c}
\inferencerule{}
              {[\ket{\varphi}, \app{\abstr{x}{\sigma}{M}}{V}] 
                    \redprob{1} [\ket{\varphi},M\{x:=V\}]}
              {\BETAV} \\ \\
\inferencerule{[\ket{\varphi},N] 
                    \redprob{p} [\ket{\varphi'},N']}
              {[\ket{\varphi}, \app{M}{N}] 
                    \redprob{p} [\ket{\varphi'},\app{M}{N'}]}
              {\BETA} \\ \\
\inferencerule{[\ket{\varphi},N] 
                    \redprob{p} [\ket{\varphi'},N']}
              {[\ket{\varphi}, \app{M}{N}] 
                    \redprob{p} [\ket{\varphi'},\app{M}{N'}]}
              {\APP} \\ \\
\inferencerule{[\ket{\varphi},M] 
                    \redprob{p} [\ket{\varphi'},M']}
              {[\ket{\varphi}, \app{M}{V}] 
                    \redprob{p} [\ket{\varphi'},\app{M'}{V}]}
              {\APPC} \\ \\
\inferencerule{}
              {[\ket{\varphi}, \ifthls{\true}{M}{N}] 
                    \redprob{1} [\ket{\varphi},M]}
              {\IFT} \\ \\
\inferencerule{[\ket{\varphi},P] 
                    \redprob{p} [\ket{\varphi'},P']}
              {[\ket{\varphi}, \ifthls{P}{M}{N}] 
                    \redprob{p} [\ket{\varphi},\ifthls{P'}{M}{N}]}
              {\IF} \\ \\
\inferencerule{}
              {[\ket{\varphi}, \ifthls{\false}{M}{N}] 
                    \redprob{1} [\ket{\varphi},N]}
              {\IFF} \\ \\
\inferencerule{i \in \{1,2\}}
              {[\ket{\varphi}, \proj{i}{\couple{V_1,V_2}}]
                    \redprob{1} [\ket{\varphi},V_i]}
              {\PROJ i} \\ \\
\inferencerule{[\ket{\varphi},M] \redprob{p} [\ket{\varphi'},M']}
              {[\ket{\varphi}, \couple{M,N}]
                    \redprob{p} [\ket{\varphi'},\couple{M',N}]}
              {\CPLL} \et \et
\inferencerule{[\ket{\varphi},N] \redprob{p} [\ket{\varphi'},N']}
              {[\ket{\varphi}, \couple{V,N}]
                    \redprob{p} [\ket{\varphi'},\couple{V,N'}]}
              {\CPLR} \\ \\ 

\inferencerule{}
              {[\ket{\varphi}, \app{\phase}{q_{i}}]
                    \redprob{1} [\phase^{i}(\ket{\varphi}),q_{i}]}
              {\PHASE} \et \et
\inferencerule{}
              {[\ket{\varphi}, \app{\hadamard}{q_{i}}]
                    \redprob{1} [\hadamard^{i}(\ket{\varphi}),q_{i}]}
              {\HADA} \\ \\
\inferencerule{}
              {[\alpha \ket{\varphi_{\false}} + \beta \ket{\varphi_{\true}}, \app{\meas}{q_{i}}]
                    \redprob{\vert \alpha \vert^2} [\ket{\varphi_{\false}},\true]}
              {\MESF} \\ \\
\inferencerule{}
              {[\alpha \ket{\varphi_{\false}} + \beta \ket{\varphi_{\true}}, \app{\meas}{q_{i}}]
                    \redprob{\vert \beta\vert^2} [\ket{\varphi_{\true}},\false]}
              {\MEST} \\ \\

\inferencerule{}
              {[\ket{\varphi}, \app{\cnot}{\couple{q_i}{q_j}}]
                    \redprob{1} [\cnot^{i,j}(\ket{\varphi}),\couple{q_i,q_j}]}
              {\CNOT} \\ \\

\end{array}$$

\end{definition}

   In rules $\MEST$ and $\MESF$, let $\ket{\varphi} = \alpha \ket{\varphi_{\false}} +
   \beta \ket{\varphi_{\true}}$ be normalized with 
    $$\begin{array}{c}
             \ket{\varphi_{\true}} = \sum_i=1^{n} \alpha_{i} \ket{\phi_i^{\true}} 
                  \otimes \ket{\true} \otimes \ket{\psi_i^{\true}} \\
             \ket{\varphi_{\false}} = \sum_i=1^{n} \beta_{i} \ket{\phi_i^{\false}} 
                  \otimes \ket{\false} \otimes \ket{\psi_i^{\false}}
      \end{array}$$. 
   where $\ket{\true}$ and $\ket{\false}$ is the ith quantum  bit.  

   We say that the set of rules containing $\BETA$, $\BETAV$, $\APP$,
$\APPC$, $\APPV$, $\IF$, $\IFF$, $\IFT$, $\PROJ i$, $\CPLL$, $\CPLR$ is the purely
functional part of $\lambdaq$, the other rules are the quantum part
of $\lambdaq$. 

   Based on this reduction rules one can define reachable states, by
considering the reflexive-transitive closure of $\redprob{p}$. One has
to compose probabilities along a reduction path. Therefore
$[\ket{\varphi'},M']$ is reachable from  $[\ket{\varphi'},M']$, if 
there is a non zero probability path between those states. More
precisions can be found in \cite{SelVal05}. 

   Computations of a $\lambdaq$ term are done from an initial state
where all registers are set to $\ket{\false}$: $\ket{\varphi_{\false}} = \ket{\false} 
       \stackrel{n-1}{\overbrace{\otimes \ldots \otimes}} 
    \ket{\false}$

  \begin{proposition}[Subject Reduction]
    \label{prop:subject_reduction}
    Let $\Gamma,\Lambda \vdash M:\tau$ and $M \redprob{p} M'$, 
    then $\Gamma,\Lambda \vdash M':\tau$
  \end{proposition}

  \begin{proof}
    From the typing point of view $\lambdaq$ is nothing more than a
    simply typed $\lambda$-calculus with constants for quantum bits
    manipulations. Note that $\phase,\hadamard,\cnot$ act as identity
    functions (from the strict $\lambda$-calculus point of view). The
    measurement is simple to deal with since it only returns constant
    (hence typable in any contexts).   
  \end{proof}

\section{Entanglement logic for $\lambdaq$}
  \label{sec:entanglement}

We present a static analysis for the study of the entanglement
relation during a quantum computation. The idea that we follow in this
paper is to adapt the work \cite{BerHonYos05} to the quantum
setting. The logic is in the style of Hoare \cite{Hoa69} and leads to the
following notation:
  $$\assertion{C}{M:^{\Gamma;\Lambda,}u}{C'} $$
where $C$ is a precondition, $C'$ is a post-condition, $M$ is the
subject and $u$ is its anchor (the name used in $C'$ to denote $M$
value).  Informally, this judgment can be red: if $C$ is
satisfied, then after the evaluation of $M$, whose value is denoted by
$u$ in $C'$, $C'$ is satisfied. $\Gamma;\Lambda$ is the typing context
of $M$ and $\Delta$ is the anchor typing context : it is used in order
to type anchors within assertions. Indeed, anchors denote terms and
have to be typed.

Since we are interested in separability analysis,
assertions state whether two quantum bits are entangled or
not. Moreover, since separability is uncomputable (it trivially reduces to the
halt problem since on can add $\cnot(q_i,q_j)$ as a last line of a
program in such a way that $q_i$ and $q_j$ are entangled iff the
computation stops), assertions are safe approximations: if an assertion
state that two quantum bits are separable then they really are,
whereas if two quantum bits are stated entangled by an assertion, it
is possible that in reality they are not.  

  \subsection{Assertions}
    \label{subsec:assretions}
 
  \begin{definition}
    \label{def:assertions}
    Terms and assertions are defined by the following grammar:
   $$\begin{array}{lrcl}
          e,e' &::=& u \mid q_i\mid \couple{e,e'} \mid \pi_i(e)\\ \\
          C,C' &::=& u \entangle v 
                                 \mid \pure e \mid  e = e' \\
               &   &  \neg C \mid C \vee C' \mid C \wedge
                    C' \mid C \implies C' \mid \forall u.C \mid
                    \exists u.C \\
               &   & \assertion{C}{\evaluation{e_1}{e_2}{e_3}}{C'}
      \end{array}$$

    Where $u,v$ are names from a countable set of anchor names. 
  \end{definition}

The idea behind assertions is the following: every subterm of a
program is identified in assertions by an anchor, which is simply a
unique name. The anchor is the logical counterpart of the
program. Note that the name of quantum bits are considered as ground terms.

Assertion $u \entangle v$ means that the quantum bit
identified by $u$ is possibly entangled with $v$. Notice that $\neg u
\entangle v$ means that it is sure that $u$ and $v$ are
separable. $\pure u$ means that it is for sure that the quantum bit is
in a base state (it can be seen as $\alpha \ket{b}$ where $b$ is
either $\true$ or $\false$). Thus $\neg \pure u$ means that $u$ may
not be in a base state (here the approximation works the other
around).  Assertion $\assertion{C}{\evaluation{e_1}{e_2}{e_3}}{C'}$ is
used to handle higher order functions. It is the evaluation
formula. $e_3$ binds its free occurrences in $C'$. following
\cite{BerHonYos05}, $C,C'$ are called internal pre/post
conditions. The idea is that invocation of a function denoted by $e_1$
with argument $e_2$ under the condition that the initial assertion $C$
is satisfied by the current quantum state evaluates in a new quantum
state in which $C'$ is satisfied. $C'$ describes the new entanglement
and purity relations.

The other assertions have their standard first order logic
meaning. Notice that in $\forall$ and $\exists$ binder are only meant
to be used on quantum bits. That is $\forall u.C$ means that $u$ is
either of the form $q_i$ or of the form $x$, with $x$ of type $\qbits$
but cannot be of the form $\couple{e,e'}$. 

  In the following we $\vrai$ (resp. $\faux$) for the following
tautology (resp. antilogy) $u=u$ (resp $\neg (u=u)$).

  \begin{definition}[Assertion typing]
  \label{def:asertiontype}
  \begin{itemize}
    \item A logical term $t$ is well typed of type $\tau$, written
    $\Gamma; \Lambda; \Delta \vdash t:\tau$ if it can be derived from
    the following rules:
    $$\begin{array}{c} 
       \inferencerule{(u:\tau) \in \Gamma; \Lambda; \Delta}
                     {\Gamma; \Lambda; \Delta \vdash u : \tau}
                     {\TTAX} \\ \\
       \inferencerule{}
                     {\Gamma; \Lambda; \Delta \vdash q_i : \qbits}
                     {\TTQ} \\ \\
       \inferencerule{\Gamma; \Lambda; \Delta \vdash e : \tau \et
                      \Gamma; \Lambda; \Delta \vdash e' : \tau'}
                     {\Gamma; \Lambda; \Delta \vdash \couple{e,e'} :
                       \tau \otimes \tau'}
                     {\TTPROD} \\ \\
       \inferencerule{\Gamma; \Lambda; \Delta \vdash u : \tau_1 \otimes \tau_2}
                     {\Gamma; \Lambda; \Delta \vdash \pi_{i}(u) : \tau_{i}}
                     {\TTPI} \\               
      \end{array}$$
      with $ i \in \{1,2\}$. 

    \item An assertion $C$ is well typed under context
    $\Gamma; \Lambda; \Delta$ written $\Gamma; \Lambda; \Delta \vdash
    C$ if it can be derived from the following rules: 
    $$\begin{array}{c}
       \inferencerule{\Gamma; \Lambda; \Delta \vdash e: \qbits \et
                      \Gamma; \Lambda; \Delta \vdash e': \qbits }
                     {\Gamma; \Lambda; \Delta \vdash e \entangle e'}
                     {\ATEN} \\ \\
       \inferencerule{\Gamma; \Lambda; \Delta \vdash e: \qbits}
                     {\Gamma; \Lambda; \Delta \vdash \pure e}
                     {\ATPU} \\ \\
       \inferencerule{\Gamma; \Lambda; \Delta \vdash e:\tau \et
                       \Gamma; \Lambda; \Delta \vdash e':\tau}
                     {\Gamma; \Lambda; \Delta \vdash e=e'}
                     {\ATEQ} \\ \\
       \inferencerule{\Gamma; \Lambda; \Delta \vdash C}
                     {\Gamma; \Lambda; \Delta \vdash \neg C}
                     {\ATNEG} \\ \\
       \inferencerule{\Gamma; \Lambda; \Delta \vdash C \et 
                      \Gamma; \Lambda; \Delta \vdash C'}
                     {\Gamma; \Lambda; \Delta \vdash C \wedge C'}
                     {\ATAN} \\ \\
       \inferencerule{\Gamma; \Lambda; \Delta \vdash C \et 
                      \Gamma; \Lambda; \Delta \vdash C'}
                     {\Gamma; \Lambda; \Delta \vdash C \vee C'}
                     {\ATOR} \\ \\ 
       \inferencerule{\Gamma; \Lambda; \Delta \vdash C \et 
                      \Gamma; \Lambda; \Delta \vdash C'}
                     {\Gamma; \Lambda; \Delta \vdash C \implies C'}
                     {\ATIMP} \\ \\ 
       \inferencerule{\Gamma; \Lambda; \Delta,u:\qbits \vdash C}
                     {\Gamma; \Lambda; \Delta \vdash \forall u.C}
                     {\ATFORALL} \\ \\ 
       \inferencerule{\Gamma; \Lambda; \Delta,u:\qbits \vdash C}
                     {\Gamma; \Lambda; \Delta \vdash \exists u.C}
                     {\ATEXISTS} \\ \\
       \inferencerule{\Gamma; \Lambda,\Lambda'; \Delta \vdash C  \et
                      \Gamma; \Lambda,\Lambda'; \Delta,e3 : \tau \vdash C' \et
                      \superpose{\Gamma; \Lambda'; \Delta \vdash e2 : \sigma}
                                {\Gamma; \Lambda; \Delta \vdash e1 : \sigma \to \tau}
                               }
                     {\Gamma; \Lambda,\Lambda'; \Delta \vdash \{C\}
                            \evaluation{e1}{e2}{e3}\{C'\}}
                     {\ATHIGORD} 
      \end{array}
    $$
  \end{itemize}
  \end{definition}

  Assertion typing rules may be classified in two categories. The
first one is the set of rules insuring correct use of names with
respect to the type of the term denoted by them. It is done by rules
$\ATEN$ $\ATPU$ $\ATEQ$ $\ATFORALL$ $\ATEXISTS$ and $\ATHIGORD$. The
second set of rules is used to structurally check formulas: $\ATNEG$
$\ATAN$ $\ATOR$, and $\ATIMP$. 

  \subsection{Semantics}
    \label{subsec:semantics}

   We now formalize the intuitive semantics of assertions. For this, we
abstract the set of quantum bits to an abstract quantum state. The
approximation (we are conservative in saying that two quantum bits are
entangled and in stating the non-purity of a quantum bits) is done at
this level. It means that for a given quantum state there are several
abstract quantum state acceptable. For instance stating that all
quantum bits are entangled, and not one of them is in a base state, which is
tautological, holds as an acceptable abstract quantum state for any actual
quantum state. The satisfaction of an assertion is done relatively to
the abstract operational semantics. We develop an abstract operational
semantics in order to abstractly execute $\lambdaq$ programs.

     \subsubsection{Abstract quantum state and abstract operational
       semantics} 
  
  Let the fixed set of $n$ quantum bits be named $S$ in the following of
  this section. Let also suppose that the quantum state of $S$ is
  described by $\ket{\varphi}$ a normalized vector of $\C^{2}$.  

  \begin{definition}[Abstract quantum state]
    \label{def:abstract}
    An abstract quantum state of $S$ (AQS for short) is a tuple
    $A=(\relentangle,\purestate)$ where $\purestate \subseteq S$ and
    $\relentangle$ is a partial equivalence relation on $(S\setminus
    \purestate) \times (S \setminus \purestate)$.
  \end{definition}

  Relation $\relentangle$ is a PER since it describes an approximation
  of the entanglement relation and there is not much sens in talking
  about the entanglement of a quantum bit with itself. Indeed because
  of the no-cloning property it is not possible to have programs $p :
  \qbits\times \qbits \to \tau$ requiring two non entangled quantum
  bits and to type $(p \ \ \couple{q_i,q_i})$. 

  The equivalence class of a quantum bit $q$ with relation to an
  abstract quantum state $A=(\relentangle,\purestate)$ is written 
  $\overline{q}^{A}$.

  \begin{definition}[AQS and quantum state adequacy]
    \label{def:adequacy}
    Let $S$ be described by $\ket{\varphi}$ and
    $A=(\relentangle,\purestate)$ an AQS of $S$. $A$ is adequate with
    regards to $\ket{\varphi}$, written $A \models \ket{\varphi}$, iff
    for every $x,y \in S$ such that $(x,y) \not \in \relentangle$ then
    $x,y$ are separable w.r.t. $\ket{\varphi}$ and for every $x \in
    \purestate$ then the measurement of $x$ is deterministic.
  \end{definition}

  Suppose that $S=\{q_1,q_2,q_3\}$ and 
  $\ket{\varphi} = 1/\sqrt(2) (\ket{\false}+\ket{\true}  \otimes 
                   1/\sqrt(2) (\ket{\false}+\ket{\true}  \otimes
                   \ket{\true}$
  then: 
  \begin{itemize}
    \item $A=(\{(q_1,q_2),(q_2,q_1)\},\{q_3\})$ 
    \item $A'=(\{(q_1,q_2),(q_2,q_1),(q_2,q_3),(q_3,q_2),(q_3,q_1),(q_1,q_3)\},\emptyset)$
  \end{itemize}
  are such that $A\models \ket{\varphi}$ and
  $A'\models\ket{\varphi}$. On the other hand:
  \begin{itemize}
    \item $B=(\{(q_1,q_2),(q_2,q_1)\},\{q_2,q_3\})$
    \item $B'=(\emptyset,\{q_3\})$
  \end{itemize}
  are not adequate abstract quantum states with relation to $\ket{\varphi}$. 

   We now give a new operational semantics of $\lambdaq$ terms based
on abstract quantum states transformation.  

  \begin{definition}[Abstract operational semantics]
    \label{def:abstropsem}
    We define an abstract operational semantics of a term $M$ such that 
    $\Gamma;\Lambda\vdash M : \tau$ between AQS as :
      $$[A,M]\redprobabstr^{\Gamma,\Lambda}[A',M']$$

    We often write $\redprobabstr$ instead of
    $\redprobabstr^{\Gamma,\Lambda}$ when typing contexts play no role
    or can be inferred from the context. 
    
    Reduction rules are the same ones as those of definition
    \ref{def:quantumreduction} for the functional part of the calculus
    where the quantum state is replaced with an
    abstract state. We have the following rules for the quantum actions:
    $$\begin{array}{ccc}
      \inferencerule{}
              {[(\relentangle,\purestate), \app{\phase}{q_{i}}]
                    \redprobabstr [(\relentangle,\purestate),q_{i}]}
              {\PHASEABS} \\ \\
       \inferencerule{}
              {[(\relentangle,\purestate), \app{\hadamard}{q_{i}}]
                    \redprobabstr [(\relentangle,\purestate \setminus \{q_{i}\}),q_{i}]}
              {\HADAABS} \\ \\
       \inferencerule{}
              {[(\relentangle,\purestate), \app{\meas}{q_{i}}]
                    \redprobabstr [(\relentangle \setminus q_i,\purestate \cup \{q_i\}),\random]}
              {\MESABS} \\ \\
     \inferencerule{}
              {[(\relentangle,\purestate), \app{\cnot}{\couple{q_i,q_j}}]
                    \redprobabstr [(\relentangle,\purestate),\couple{q_i,q_j}]}
              {\CNOTABSONE} \mbox{if $q_i \in \purestate$}\\ \\
     \inferencerule{}
              {[(\relentangle,\purestate), \app{\cnot}{\couple{q_i,q_j}}]
                    \redprobabstr [(\relentangle \cdot q_i \entangle
                    q_j , \purestate\setminus\{q_i,q_j\}),\couple{q_i,q_j}]}
              {\CNOTABSZERO} \mbox{if $q_i \not \in \purestate$}
      \end{array}$$

      Where $\random$ is non deterministically $\true$ or $\false$,
      $\relentangle \setminus q_i$ is the equivalence relation such
      that if $(x,y)\in relentangle$ and $x\not=q_i$ or exclusive $y
      \not=q_i$ then $(x,y) \in \relentangle \setminus q_i$ otherwise
      $(x,y) \not \in \relentangle \setminus q_i$, and where
      $\relentangle \cdot q_i \entangle q_j $ is the equivalence
      relation $\relentangle$ in which the equivalence classes of
      $q_i, q_j$ have been merged together.
  \end{definition}

  Note that this abstract semantics is not deterministic since it
non deterministically gives $\true$ or $\false$ as result of a
measure. Its correctness can hurt the intuition since the measurement
of a quantum bit in a base state, say $\ket{\true}$, can never produce
$\ket{\false}$. Note also that since our system is normalizing the
number of all possible abstract executions is finite. Hence, computable.

  \begin{definition}[Abstract program semantics]
    \label{def:}
    Consider an AQS $A$, the semantics of program
    $\Gamma;\Lambda\vdash M:\tau$ under $A$,
    written $\progsem{M}{A}^{\Gamma;\Lambda}$, is the set of $A'$ such
    that $[A,M] \redprobabstr^{*}
    [A',V]$ where $V$ is a value. 
  \end{definition}

  Notice that the abstract semantics of a program is a collecting
  semantics. It may explore branches
that are never going to be used in actual computation. Indeed in the
operational semantics measurement gives a non deterministic
answer. Nevertheless, correctness is ensured by the if judgment rules
(see rule $\JDGIF$ in definition \ref{def:proofrules}).   

  \begin{proposition}
    \label{prop:abstrexec}
    Let $A \models \ket{\varphi}$, $\Gamma;\Lambda\vdash M:\tau$.
    Suppose that $[\ket{\varphi},M] \redprob{\gamma}^{*} [\ket{\varphi'},V]$
    then there exists $A' \models \ket{\varphi'}$ such that
    $[A,M]\redprobabstr^{*}[A',V]$.  
  \end{proposition}

  \begin{proof}
    The proof is done by induction on the number of steps of the
    reduction between $[\ket{\varphi},M]$ and $[\ket{\varphi'},V$. The
    proposition is clearly true if there is $0$ step since $M=V$,
    $\varphi=\varphi'$ and $A'=A$ proves the result. 

    Now consider the last rule used. If this rule is one of the
    purely functional part of the calculus (see
    def. \ref{def:quantumreduction}) the proposition follow directly
    from the induction hypothesis since the AQS is not changed. We
    thus have the following possibilities for the last rule:
    \begin{itemize}
      \item It is $\PHASEABS$: If the qbit $q$ on which phase is applied
        is a base state it can be written $\alpha \ket{l}$ with $l$
        being either $\true$ or $\false$. Thus $\phase q =
        \exp^{i\pi/4}\alpha$, thus still a base state. Hence
        $\purestate$ remains unchanged. 
      \item It is $\HADAABS$: if $(\relentangle,\purestate) \models \varphi$, then 
               $(\relentangle,\purestate\setminus\{q_i\}) \models
               (\hadamard_i \ \ \ket{\varphi})$ because of definition
               \ref{def:adequacy} since in $(\hadamard_i \ \
               \ket{\varphi})$, any $q_j$ is in a non base state only if
               it is in a non base state in $\ket{\varphi}$. 

      \item It is $\MESABS$: After the measure the qubit vanishes. Moreover
            concrete measure probabilistically produces $\true$ or
            $\false$. Regarding the concrete result one can choose the
            appropriate value as result of the abstract measure,
            moreover the measured qubit is in a base state (hence the 
            $\purestate\cup\{q_i\}$).   

      \item It is $\NEWABS$: then by definition $\ket{\varphi'}=\ket{\true} \otimes 
            \ket{\varphi}$, hence quantum in a base state in $\varphi$ remain in a base state in 
            $\varphi'$, moreover the new qubit is in a base state. 

      \item It is $\CNOTABSZERO$: If the two qubits $q_i=\alpha \ket{l}, q_j=\beta \ket{l'}$ 
            are in a base state then 
            \begin{itemize}
              \item If $l=\true$ then $\cnot(\alpha\ket{\true}\otimes\beta\ket{l'})=
                    \alpha\ket{\true}\otimes\beta\ket{\neg l'}$
              \item If $l=\false$ then $\cnot(\alpha\ket{\false}\otimes\beta\ket{l'})=
                    \alpha\ket{\false}\otimes\beta\ket{l'}$
            \end{itemize} 
            in both cases we obtain two separable qubits.  

            If only $q_i= \alpha'\ket{l}$ is in a base state and 
           $q_j=\alpha \ket{\false}+\beta \ket{\true}$ is not.
           \begin{itemize}
             \item If $l=\true$ then $\cnot(\alpha \ket{\true}\otimes
               \alpha \ket{\false}+\beta \ket{\true})
               = \alpha'\ket{\true} \otimes \beta \ket{\true}+\alpha \ket{\false}$ 
             \item If $l=\false$ then $\cnot(\alpha'\ket{\true}
               \otimes \alpha \ket{\true}+\beta \ket{\false}) = \alpha'\ket{\true}
               \otimes \alpha \ket{\true}+\beta \ket{\false})$
           \end{itemize}
           here also we obtain two separable qubits. Moreover in all
           cases $q_i$ remains in a base state. 
  
      \item It is $\CNOTABSONE$: The property follows from induction
        hypothesis and from the fact that $\relentangle$ and
        $\purestate$ are safe approximations. 
    \end{itemize}
  \end{proof}
 
\subsubsection{Semantics of entanglement assertions}

  We now give the semantics of a well typed assertion with relation to
a concrete quantum state. It is done via an abstract quantum state
which is adequate with regards to the concrete quantum state. The idea
is as follows: if $\ket{\varphi} \models A$, and if $\Gamma;\Lambda;
\Delta \vdash C$ then we define the satisfaction relation
$\model^{\Gamma;\Lambda;\Delta} \models C$, which states that under a
proper model depending on the typing context, then $C$ is
satisfied. Basically it amounts to check two properties : whether or
not two quantum bits are in the same entanglement equivalence class
and whether or not a particular quantum bit is in base state. 

  \begin{definition}[Abstract observational equivalence]
    \label{def:equobs}
    Suppose that $\Gamma;\Lambda \vdash M,M' : \tau$. $M$ and $M'$ are
    observationally equivalent, written $M \equobs_{A}^{\Gamma,\Lambda} M'$, if and
    only if for all context $C[.]$ 
    such that $\cdot;\cdot \vdash C[M], C[M']:\bits$ and for all AQS $A$ we have
        $$\progsem{C[M]}{A}^{\Gamma,\Lambda}=\progsem{C[M']}{A}^{\Gamma,\Lambda}$$
    The equivalence class of $M$ is denoted by
    $\equivclass{M}_{A}^{\Gamma,\Lambda}$, by extension we say that
    the type of this equivalence class is $\tau$. 
  \end{definition}
  
  \begin{definition}[Abstract values]
    \label{def:abstractvalue}
    In assertion typing context $\Gamma; \Lambda; \Delta$, an abstract value 
    $v^{\Gamma;\Lambda;\Delta}_{A,tau}$ of type
    $\tau$, where $\tau \not = \sigma \otimes \sigma'$, with relation
    to context $\Gamma;\Lambda;\Delta$ and AQS $A=(\relentangle,\purestate)$ is:  
    \begin{itemize}
      \item An equivalence class of type $\tau$ for
        $\equobs_{A}^{\Gamma,\Lambda}$, if $\tau \not = \qbits$. 
      \item a pair $(C,b)$ formed by an equivalence class $C$ of
        $\relentangle$ and a boolean $b$ (the idea being that if $b$
        is true then the denoted qubit is in $\purestate$).
    \end{itemize}
 
    If $\tau = \sigma' \otimes \sigma''$, then
    $v^{\Gamma;\Lambda;\Delta}_{A,\tau}$ is a pair $(v',v'')$ formed by abstract
    values of respective types $\sigma',\sigma''$.  

    The set of abstract values under an AQS $A$, typing context 
    $\Gamma;\Lambda;\Delta$ and for a type $\tau$ is written 
    $\abstrvalset{A,\tau}^{\Gamma;\Lambda;\Delta}$.   
  \end{definition}

  Abstract values are used to define the interpretation of
free variables. Since in a Given an assertion typing context $\Gamma; 
\Lambda; \Delta$ more than one type may occur we need to consider collections of
abstract values of the different types that occur in $\Gamma;\Lambda; \Delta$
: we write $\abstrvalset{\Gamma;\Lambda;\Delta}$ the disjoint union of all
$\abstrvalset{\tau}^{\Gamma;\Lambda;\Delta}$ for every $\tau$ in
$\Gamma; \Lambda; \Delta$.

 \begin{definition}[Models]
   \label{def:models} A $\Gamma;\Lambda;\Delta$ model is a tuple
   $\model^{\Gamma; \Lambda; \Delta} =\langle A, \interpretation
   \rangle$, where $A$ is an AQS, $\interpretation$ is a map from
   variables defined in $\Gamma;\Lambda;\Delta$ to  
   $\abstrvalset{\Gamma;\Lambda;\Delta}$. 
 \end{definition}

 In order to deal with evaluation and quantified formulas we need to define a notion
of model extension. 

 \begin{definition}[Model extensions]
   \label{def:modelextension}
   Let $\model^{\Gamma;\Lambda;\Delta} =\langle A,\interpretation\rangle$ be a model, 
   then the model $\model'$ written $\model \cdot
   x:v=\langle A,\interpretation'\rangle$, where 
   $v \in \abstrvalset{A,\tau}^{\Gamma;\Lambda;\Delta}$ is defined as follows:
   \begin{itemize}
     \item the typing context of $\model'$ is $\Gamma; \Lambda; \Delta,x:\tau$. 
     \item If the type of $x$ is $\tau = \sigma \otimes \sigma'$, then
       $v$ is a couple made of abstract values $V',v''$ of respective
       type $\sigma, \sigma'$.  
     \item If the type of $x$ is $\qbits$: if $v=(C,\true)$ then $A'=(\relentangle \cup C, 
           \purestate'\cup\{x\})$, otherwise if $v=(C,\false)$ then $A'=(\relentangle \cup C, 
           \purestate')$. 
     \item If the type of $x$ is $\sigma \not = \qbits$, then: $\interpretation'(y) = 
           \interpretation(y)$ for all $x \not = y$ and $\interpretation'(x)=v$ 
   \end{itemize}   
  \end{definition}

  We now define term interpretation. It is standard and amounts to an
interpretation of names into abstract values of the right type. 

  \begin{definition}[Term interpretation]
   \label{def:terminterpretations}
   Let $\model^{\Gamma,\Lambda}=\langle A,\interpretation,\tau \rangle$ be a model, the
   interpretation of a term $u$ is defined by:
   \begin{itemize}
     \item $\interpret{u}_{\model}=\interpretation(u)$ if the type of $u$ is not $\qbits$.
     \item $\interpret{q_i}_{\model}=(\overline{q_i}^{A},b_{i}^{A})$,
       where $b_i^A$ is true iff $q_i \purestate$ with $A=\couple{\relentangle,\purestate}$.
     \item $\interpret{\couple{e,e'}}_{\model} = \couple{\interpret{\model}_{A},
            \interpret{e'}_{\model}}$
   \end{itemize} 
  \end{definition}

  \begin{definition}[Satisfaction]
   \label{def:satisfaction}
   The satisfaction of an assertion $C$ in the model $\model =
   \langle A,\interpretation \rangle$, is written $\model \models C$, is
   inductively defined by the following rules:
   \begin{itemize}
     \item $\model \models u \entangle v$ if
       $(\pi_1(\interpret{u}_{\model}),\pi_1(\interpret{v}_{model})) \in \relentangle_{A}$.
     \item $\model \models \pure u$  if $\pi_2(\interpret{u}_{\model})$ is true.
     \item $\model \models e_1 = e_2$ if $\interpret{e_2}_A = \interpret{e_1}_A$.
     \item $\model \models \neg C$ if $\models$ does not satisfy $C$.
     \item $\model \models C \vee C'$ if $\model \models C$ or $\model \models C'$. 
     \item $\model \models C \wedge C'$ if  $\model \models C$ and $\model \models C'$.
     \item $\model \models C \implies C'$ if $\model \models C$
           implies $\model \models C'$.
     \item $\model \models \forall u.C$ if for all model
       $\model'=\model\cdot u.v$, one has $\model' \models C$. 
     \item $\model \models \exists u.C$ if there is an abstract value
       $v$ such that if $\model'=\model\cdot u.v$, one has $\model' \models C$. 
     \item $\model \models
       \assertion{C}{\evaluation{e_1}{e_2}{e_3}}{C'}$ if for all models
       $\model'^{\Gamma;\Lambda;\Delta}=\langle A', \interpretation'\rangle$ such that
       $\model'^{\Gamma;\Lambda;\Delta'} \models C$, with the
       following conditions: 
       $\Gamma;\Lambda;\Delta \vdash e_1 : \sigma \to \tau$, and
       $\Gamma; \Lambda; \Delta \vdash e_2: \sigma$ such that for all terms 
       $t_1 \in \interpret{e_1}_{\model'}, t_2 \in \interpret{e_2}_{\model'}$ one has
       \begin{itemize}
         \item $[A,\app{t_1}{t_2}]\redprobabstr^{*}[A',V]$
         \item we have two sub-cases:
           \begin{enumerate}
                  \item $\tau$ is $\qbits$ and $V=q_i$ and
                    $\model'= \model \cdot e_3:(\overline{q_i}^{A'},
                    q_i \in \purestate_{A'})$ 
                  \item $\tau$ is not $\qbits$ and  $\model' \cdot e_3:
                 \equivclass{V}_{A}^{\Gamma;\Lambda;\Delta,\tau} \models C'$
           \end{enumerate}
       \end{itemize}
   \end{itemize} 
  \end{definition}

  \subsection{Judgments and proof rules }
    \label{subsec:judgements}

   We now give rules to derive judgments of the form 
$\judge{C}{M}{\Gamma;\Lambda;\Delta;\tau}{u}{C'}$. Those judgments
bind $u$ in $C'$, thus $u$ cannot occur freely in $C$. There are two kinds
of rules: the first one follow the structure of $M$, the second one
are purely logical rules.   
 
   \begin{definition}[Language rules]
     \label{def:proofrules}
      Let $\Gamma;\Lambda \vdash M:\tau$, we define the judgment
      $\judge{C}{M}{\Gamma;\Lambda;\Delta;\tau}{u}{C'}$ inductively as follows: 

     $$\begin{array}{cc}
        \inferencerule{\judge{C \wedge \pure u}{N}
                             {\Gamma;\Lambda;\Delta;\qbits\otimes \qbits}{\couple{u,v}}{C'} }
                      {\judge{C \wedge \pure u}{(\cnot \  N)}
                             {\Gamma;\Lambda;\Delta;\qbits\otimes\qbits}{\couple{u,v}}{C'}}
                      {\JDGCNOTONE} \\ \\ 
        \inferencerule{\judge{C}{N}
                             {\Gamma;\Lambda;\Delta;\qbits \otimes \qbits}{\couple{u,v}}{C'}}
                      {\judge{C}{(\cnot \  N)}
                             {\Gamma;\Lambda;\Delta;\qbits
                               \otimes \qbits}
                             {\couple{u,v}}{C'\wedge u \entangle v}}
                      {\JDGCNOTTWO} \\ \\
        \inferencerule{\judge{C}{N}
                             {\Gamma;\Lambda;\Delta;\qbits}{v}{C'}}
                      {\judge{C}{(\hadamard \ N)}
                             {\Gamma;\Lambda;\Delta;\qbits}{v}{C'[\neg \pure v]}}
                      {\JDGHAD} \\ \\
        \inferencerule{\judge{C}{N}
                             {\Gamma;\Lambda;\Delta;\qbits}{u}{C'}}
                      {\judge{C}{(\phase \ N)}
                             {\Gamma;\Lambda;\Delta,\qbits}{u}{C'}}
                      {\JDGPHASE}\\ \\ 
        \inferencerule{}
                      {\judge{C[u/x]}{x}
                             {\Gamma;\Lambda;\Delta,u:\tau;\tau}{u}{C}}
                      {\JDGVAR} \\ \\
        \inferencerule{c \in \{\true,\false\}}
                      {\judge{C}{c}
                             {\Gamma;\Lambda;\Delta,u:\bits;\bits}{u}{C}}
                      {\JDGCONST} \\ \\
          \inferencerule{\judge{C}{M}
                               {\Gamma;\Lambda;\Delta}{\Gamma;\Lambda;\Delta;\qbits}{u}{C'}}
                        {\judge{C}{\meas \ M}
                               {\Gamma;\Lambda;\Delta,v:\bits;\bits}{v}{C'[-u]\wedge \pure u}}
                        {\JDGMEAS} \\ \\
          \inferencerule{\judge{C}{M}
                               {\Gamma;\Lambda;\Delta;\bits}{b}{C_0}~~
                         \judge{C_0[\true/b]}{N}
                               {\Gamma;\Lambda;\Delta;\tau}{x}{C'}~~
                         \judge{C_0[\false/b]}{P}
                               {\Gamma;\Lambda;\Delta;\tau}{x}{C'}}
                        {\judge{C}{\ifthls{M}{N}{P}}
                               {\Gamma;\Lambda;\Delta,u:\tau;\tau}{u}{C'}}
                        {\JDGIF} \\ \\
          \inferencerule{\judge{C}{M}
                               {\Gamma;\Lambda;\Delta;\sigma \to \tau}{m}{C_0} ~~ 
                         \judge{C_0}{N}
                               {\Gamma;\Lambda;\Delta;\sigma}{n}{C_1 \wedge \{C_1\} m
                           \bullet n=u\{C'\}}}
                        {\judge{C}{\app{M}{N}}
                               {\Gamma;\Lambda;\Delta,u:\tau;\tau}{u}{C'}}
                        {\JDGAPP} \\ \\
           \inferencerule{\judge{C^{-x}\wedge C_0}{M}
                                {\Gamma;\Lambda;\Delta;\tau}{m}{C'} }
                         {\judge{C}{\abstr{x}{M}}
                                {\Gamma[-x];\Lambda[-x];\Delta,u:\sigma\to\tau;\sigma
                                  \to \tau}{u}
                               {\forall x.\{C_0\} u\bullet x=m\{C'\}}}
                         {\JDGABS} \\  \\  
         \inferencerule{\judge{C}{M}{\Gamma;\Lambda;\Delta;\tau}{m}{C_0} ~~ 
                         \judge{C_0}{N}{\Gamma;\Lambda;\Delta;\sigma}{n}{C'[m/u,n/v]}}
                        {\judge{C}{\couple{M,N}}
                               {\Gamma;\Lambda;\Delta,u:\tau,v:\sigma;\tau}{\couple{u,v}}{C']}}
                        {\JDGCPL} \\ \\
          \inferencerule{\judge{C}{M}
                               {\Gamma;\Lambda;\Delta;\tau_1\otimes\tau_2}{m}{C'[\pi_{i}(m)/u]} \et 
                                i \in \{1,2\}}
                        {\judge{C}{\proj{i}{M}}
                               {\Gamma;\Lambda;\Delta u:\tau_i;\tau_i}{u}{C'}}
                        {\JDGLET i}
         \end{array}$$

    Where in rule $\JDGHAD$, if there exists $C''$ such that 
   $C''\wedge \pure u \equiv C'$ the assertion $C'[\neg \pure v]$ is 
   $C''\wedge \neg \pure u$ otherwise it is $C' \neg \pure u$. In
   $\JDGMEAS$, the assertion $C'[-u]$ is $C'$ where all assertions
   containing $u$ have been deleted. In $\JDGABS$, $C^{-x}$ means that
   $x$ does not occur freely in $C$. In $\JDGVAR$, $C[u/x]$ is the
   assertion $C$ where all free occurrences of $x$ have been replaced
   by $u$. 
   \end{definition} 

 Judgment of the purely functional fragment are standard see
\cite{BerHonYos05}. We have just modified the way to handle couples in
order to ease manipulations, but we could have used projections instead
of introducing two different names. Regarding the quantum fragment,
rule $\JDGCNOTONE$ has no influences over quantum entanglement since
the first argument of the $\cnot$ is in a base state; rule
$\JDGCNOTTWO$ introduces an entanglement between the two arguments of
the $\cnot$ operator. Notice that it is not useful to introduce all
entanglement pairs introduced. Indeed, since the entanglement relation
is an equivalence relation one can safely add to judgment (see
logical rules that follow in def. \ref{def:logrules}) statements
for transitivity, reflexivity and symmetry of entanglement relation,
for instance $\forall x,y,z. x \entangle y \wedge y \entangle z \implies x
\entangle z$ for transitivity. Indeed any abstract quantum state, by definition,
validates those statements which will be implicitly supposed in the
following. As we saw in the proof of proposition \ref{prop:abstrexec}, the phase
gate does not change the fact that a quantum bit is in a base state,
whereas the Hadamard gate may make him not in a base state, hence
explaining the conclusions of rules $\JDGHAD$ $\JDGPHASE$.

  We now give purely logical rules. One may see them as an adapted
version of standard first order logic sequent calculus. 

   \begin{definition}[Logical rules]
   \label{def:logrules}
    $$\begin{array}{c}
        \inferencerule{\judge{C_0}{V}{}{u}{C_0'} \et C \vdash C_0'
          \et C_0 \vdash C'}
                      {\judge{C}{V}{}{u}{C'}}
                      {\JDGLOG} \\ \\
        \inferencerule{\judge{C}{V}{}{u}{C'}}
                      {\judge{C \wedge C_0}{V}{}{u}{C' \wedge C_0}}
                      {\JDGPROM} \\ \\
        \inferencerule{\judge{C \wedge C_0}{V}{}{u}{C'}}
                      {\judge{C}{V}{}{u}{C_0 \implies C'}}
                      {\JDGIMPEL} \\ \\
        \inferencerule{\judge{C}{M}{}{u}{C_0 \implies C'}}
                      {\judge{C \wedge C_0}{V}{}{u}{C'}}
                      {\JDGETEL} \\ \\
        \inferencerule{\judge{C_1}{M}{}{u}{C} \et \judge{C_2}{M}{}{u}{C}}
                      {\judge{C_1 \vee C_2}{M}{}{u}{C}}
                      {\JDGOUL} \\ \\
        \inferencerule{\judge{C}{M}{}{u}{C_1} \et \judge{C}{M}{}{u}{C_2}}
                      {\judge{C}{M}{}{u}{C_1 \wedge C_2}}
                      {\JDGETR} \\ \\
        \inferencerule{\judge{C}{M}{}{u}{C'^{-x}}}
                      {\judge{\exists x.C}{M}{}{u}{C'}}
                      {\JDGEXL} \\ \\
        \inferencerule{\judge{C^{-x}}{M}{}{u}{C'}}
                      {\judge{C}{M}{}{u}{\forall x.C'}}
                      {\JDGFORR} \safe{\\ \\
         \inferencerule{C \logic C_0 ~~ 
                        \judge{C_0}{M}{}{m}{C_1} ~~
                        C_1 \logic C'}
                       {\judge{C}{M}{}{m}{C'}}
                       {\JDGTRANS}}
      \end{array}$$   

      where $C \vdash C'$ is the standard first order logic proof
      derivation (see e.g. \cite{Smul68}). 
   \end{definition}

  We now give the soundness result relating 

\begin{theorem}[Soundness]
    \label{theo:validity}
    Suppose that $\judge{C}{M}{\Gamma;\Lambda;\Delta;\tau}{u}{C'}$ is
    provable. Then for all model $\model =\couple{A,\interpretation}$,
    abstract quantum state $A'$, abstract value $v$ such that 
    \begin{enumerate}
       \item $\model \models C$
       \item $[A,M] \redprobabstr^{*} [A',V]$
       \item $v \in \abstrvalset{A',\tau}^{\Gamma;\Lambda;\Delta}$
    \end{enumerate} then $\model\cdot u:v \models C'$. 
\end{theorem}
\begin{proof}
  The proof is done by induction  on judgment rules. The last judgment
  rule used can be either a logical or a language one. If it is a
  logical one, soundness follows from the soundness of first order
  logic. Observe that we have a value in the promotion rule $\JDGPROM$
  thus no reductions are possible and the soundness is vacuously valid.  

  If he last judgment rules used is a language rule, we only consider
  the quantum fragment (indeed for the functional fragment, the proof follows
  directly from \cite{BerHonYos05}), thus we have the following cases:
  \begin{itemize}
    \item $\JDGCNOTONE$, thus
      $\judge{C}{M}{\Gamma;\Lambda;\Delta;\tau}{u}{C'}$ is in facts 
      $\judge{C_1 \wedge \pure u'}{(\cnot \  N)}
      {\Gamma;\Lambda;\Delta;\qbits\otimes\qbits}{\couple{u',v'}}{C'}$.
      By induction hypothesis we know that if $\model \models C_1 \wedge
      \pure u'$, if $[A,N]\redprobabstr^{*}[A',V]$, and $v \in 
      \abstrvalset{A',\tau}^{\Gamma;\Lambda;\Delta}$, then
      $\model\cdot\couple{u',v'}:v \model C'$. We know that $V$ is a
      couple of qbits (since judgment is well typed), say
      $\couple{q_i,q_j}$. Now $[A',(\cnot \ \couple{q_i,q_j})]
      \redprobabstr [A, \couple{q_i,q_j}]$ thanks to rule
      $\CNOTABSONE$ and due to the fact that $\model \models \pure u'$.     

    \item $\JDGCNOTTWO$, thus 
      $\judge{C}{M}{\Gamma;\Lambda;\Delta;\tau}{u}{C'}$ is in facts
      $\judge{C}{(\cnot \  N)}
                             {\Gamma;\Lambda;\Delta;\qbits
                               \otimes \qbits}
                             {\couple{u',v'}}{C'\wedge u' \entangle v'}$
      we reason similarly as in previous case with the difference that 
      the last abstract operational rule used is $\CNOTABSZERO$. 

    \item $\JDGHAD$, thus 
      $\judge{C}{M}{\Gamma;\Lambda;\Delta;\tau}{u}{C'}$ is in facts
      $\judge{C}{(\hadamard \ N)}
                             {\Gamma;\Lambda;\Delta;\qbits}{u}{C'[\neg
                               \pure u]}$. By induction hypothesis we
                             know that 
     if $\model \models C$, if $[A,N]\redprobabstr^{*}[A',V]$, and $v \in 
      \abstrvalset{A',\tau}^{\Gamma;\Lambda;\Delta}$, then
      $\model\cdot\couple{u}:v \model C'$. Now because judgment is
      well typed $\tau$ is $\qbits$, and $V$ is $q_i$. Thus
      $[A,(\hadamard \ q_i)] \redprobabstr [(\relentangle,\purestate
      \setminus \{q_{i}\}),q_{i}]$, and clearly
      $\model\cdot\couple{u}:v \models \neg \pure u$, the rest is done
      by induction hypothesis.

    \item $\JDGPHASE$, thus 
      $\judge{C}{M}{\Gamma;\Lambda;\Delta;\tau}{u}{C'}$ is direct by
      induction hypothesis and considering abstract reduction rule
      $\PHASEABS$. 

    \item $JDGMEAS$, thus 
      $\judge{C}{M}{\Gamma;\Lambda;\Delta;\tau}{u}{C'}$ is in facts
      $\judge{C}{(\meas \ N)}
                             {\Gamma;\Lambda;\Delta;\qbits}{u}{C'[-u]\wedge
                               \pure u]}$. By induction hypothesis we
                             know that 
     if $\model \models C$, if $[A,N]\redprobabstr^{*}[A',V]$, and $v \in 
      \abstrvalset{A',\tau}^{\Gamma;\Lambda;\Delta}$, then
      $\model \cdot u:v \model C'$. Now because judgment is
      well typed $\tau$ is $\qbits$, and $V$ is $q_i$. Thus
      $[A,(\meas \ q_i)] \redprobabstr [(\relentangle,\purestate \cup \{q_i\}
      \setminus \{q_{i}\}), \random]$, and clearly
      $\model\cdot u:v \models \pure u$, the rest is done
      by induction hypothesis. 
  \end{itemize}
\end{proof}

\begin{example}
  \label{ex:nonlocality}
  The idea of this example is to show how the entanglement logic may
be used to analyze non local and non compositional behavior. Suppose
4 qubits, $x,y,z,t$ such that $x,t$ are entangled and
$y,z$ are entangled and $\{x,t\}$ separable from $\{y,z\}$. Now if we
perform a control not on $x,y$,  then as a side effect $z,t$ are
entangled too, even if quantum bits $x,y$ are discarded by
measurement. Thus we want to prove the following statement:

  $$\judge{\vrai}{P}{}{u}
   {\forall
     x,y,z,t.\assertion{x \entangle y \wedge z \entangle
     t}{\evaluation{u}{y,z}{v}}{x \entangle t}}$$
where $P$ is the following program 
  $$\abstr{y,z}{\letcouple{u,v}{(\cnot \ \couple{y,z})}{\couple{(\meas \
      u),(\meas \ v)}}}$$

 Then using rule $\JDGAPP$ we can derive the following judgment on
 actual quantum bits:

$$\judge{C}{(P \ \couple{q_2,q_3})}{}{\couple{u,v}}{q_1 \entangle q_4}$$

  where $C$ denotes the following assertion : $q_1 \entangle q_2 \wedge
q_3 \entangle q_4$. This judgment is remarkable in the fact that it
asserts on entanglement properties of $q_1,q_4$ while those two
quantum bits do not occur in the piece of code analyzed. 
\end{example}

\section{Conclusion}
\label{sec:conclusion}

  In this paper we have proposed a logic for the static analysis of
entanglement for a functional quantum programing language. We have
proved that this logic is safe and sound: if two quantum bits are
provably separable then they are not entangled while if they are
provably entangled they could actually be separable. The functional
language considered includes higher-order functions. It is, to our
knowledge the first proposal to do so and strictly improves over
\cite{Perdr07} on this respect. We have shown that non local behavior
can be handled by this logic. 

  Completeness of our logic remains an open issue worth of
future investigations. We also hope that this setting will allow
reasoning examples on quantum algorithms, and that it will provide a
useful help for quantum algorithms research in providing a high-level,
compositional reasoning tool. 

\bibliographystyle{abbrv}

\end{document}